\documentclass[12pt]{article}
\usepackage{graphicx,fullpage}
\usepackage{color}
\usepackage{epsfig}
\usepackage{multirow}
\usepackage{subfigure}
\usepackage{natbib,url}
\usepackage{setspace}
\usepackage{amsmath,amsfonts,amssymb,amsthm}
\usepackage[affil-it]{authblk}

\newtheorem{theorem}{Theorem}
\newtheorem{proposition}{Proposition}

\date{}
\title{Monotone false discovery rate}
\author{Joong-Ho Won%
	\thanks{E-mail: \texttt{wonj@korea.ac.kr}.}}
\affil{Korea University}

\author{Johan Lim}
\author{Donghyeon Yu}
\affil{Seoul National University}

\author{Byung Soo Kim}
\affil{Yonsei University}

\author{Kyunga Kim}
\affil{Sookmyung Womens University}

\begin{document}
\maketitle
\begin{abstract}
\noindent
This paper proposes a procedure to obtain monotone estimates of both the local and the tail false discovery rates that arise in large-scale multiple testing. The proposed monotonization is asymptotically optimal for controlling the false discovery rate and also has many attractive finite-sample properties. 

\end{abstract}
% \begin{keyword}
% adaptive decision rule, false discovery rate, empirical Bayes methods, mode matching, isotonic regression.
% \end{keyword}

%\baselineskip 24pt
%\end{frontmatter}

\section{Introduction}

The advance of modern high-throughput technologies in many scientific disciplines such as genomics and brain imaging has dramatically increased both the size and the dimension of the data and made data analysis a major challenge.
In particular, it is often required to test thousands or millions of hypotheses simultaneously when analyzing large-scale, high-dimensional data.
Unlike the case of testing a single hypothesis, type I error in multiple hypothesis testing is not uniquely defined. Traditional approaches, e.g, the family-wise error rate (FWER), are far too conservative and produce many false negatives in high-dimensional settings. 
For this reason, the concept of false discovery rate (FDR), or the expected proportion of false positives among declared positives, is introduced and now widely accepted.

The FDR is originally proposed by \citet{BenjaminiH:1995}, who develop a stepwise procedure to control the FDR. \citet{Storey:2002} proposes to estimate the FDR of a fixed rejection region and introduces the $q$-value, which is the minimum FDR level to reject the null hypothesis given observed data. 
% \textcolor{red}{Both Benjamini and Hochberg's and Storey's procedures assume independence of statistics for hypotheses.  Since this assumption does not always hold in practice,} 
Both the Benjamini-Hochberg precedure and the $q$-value assume independence among the summarizing statistics. Unfortunately the independence assumption rarely holds in practice, hence often discrepancy appears between the theoretical and the observed distributions of the summarizing statistics. For this reason,
Efron has recently introduced an empirical Bayes (EB) procedure based on a two-group mixture model \citep{Efron:2004,Efron:2007a,Efron:2007b}. 
The EB procedure uses the $z$-values instead of the $p$-values and fits them using the two-group mixture model. The EB framework introduces two variants of the FDR: the local FDR, denoted by ``fdr'', is the ratio of the null sub-density to the marginal mixture
density of the two-group model; the tail FDR, denoted by ``Fdr'', is the ratio of the null sub-survival function (tail probability) to the marginal survival function. 
The EB procedure estimates the null and the marginal mixture distributions from the data. Hence it takes into account the dependence among test statistics. The estimated null distribution is referred to as the empirical null.

The main theme of this paper is monotonicity in the FDR. Monotonicity is desirable in many settings as it maintains the order of the observed test statistics. In particular, we show that the monotonicity condition for the local FDR implies the monotone likelihood ratio condition (MLRC) of \citet{SunCai:2007},%
\footnote{%
\citet{SunCai:2007} call the condition ``SMLR,'' without spelling out what it abbreviates. Later they refer to the same condition as ``MLRC,'' while identifying that ``SMLR'' stands for ``symmetric monotone likelihood ratio'' (personal communications with Wenguang Sun, 2013).}
%(see Section \ref{sec:monotone:props} for the definition of the SMLR condition). 
under which the local FDR yields the optimal oracle decision rule. We then show that a monotone estimate of the local FDR results in a data-driven decision rule that is, under some regularity conditions, asymptotically optimal. Furthermore, we prove that a monotone estimate of the local FDR satisfies the MLRC in finite-sample settings, which by itself is desirable in practice.

Despite many attractive features of monotonicity, unfortunately, few existing procedures to estimate fdr or Fdr take monotonicity into account. \citet{Broberg:2005} investigates the use of monotone FDR in the setting that the theoretical null distribution of $p$-values is uniform on $[0,1]$. In this setting, monotonicity of fdr (resp. Fdr) is equivalent to that of the marginal density function (resp. the marginal survival function). Monotonicity is enforced by estimating the marginal density function (resp. the marginal survival function) under appropriate constraints, either parametrically or non-parametrically. A similar procedure is employed by \citet{Strimmer:2008}.
For more flexible EB procedures \citep{Efron:2007a,Efron:2007b}, however, one has to estimate both the null and the marginal distributions. We undertake to see how to impose monotonicity in this setting.

We begin with a review of the empirical Bayes theory of false discovery rate. In Section \ref{sec:monotone}, attractive statistical properties of the monotone FDR are discussed. We show that monotonicity in the local FDR is equivalent to that in the likelihood ratio of the components of the two-group mixture model, and implies that of the tail FDR. After proving the claims made above, we propose a procedure that ensures monotonicity in the estimates of the local and the tail FDRs, and that naturally leads to an adaptive decision rule using 
the monotonized estimates. In Section \ref{sec:numerical}, we conduct a numerical study to demonstrate that the monotonized FDR can improve the performance of  the FDR estimates. In Section \ref{sec:examples}, we illustrate that the proposed procedure can improve real-world data analyses. Section \ref{sec:conclusion} concludes the paper.

\section{Empirical Bayes Theory of False Discovery Rates }\label{sec:EB}

This section reviews the empirical Bayes theory of false discovery rate inference, largely developed 
by  Efron \citep{Efron:2004,Efron:2007a,Efron:2007b}.

Suppose we have a collection of $N$ hypotheses and their corresponding ``summarizing statistics'' $T_1,\ldots,T_N$. 
Assume that
the $T_i$s have a common marginal distribution whose density is of the two-group mixture form:
\begin{equation} \label{eqn:model}
f(t) = p_0 f_0(t) + p_1 f_1 (t),
\end{equation}
where $f_0 (t)$ and $f_1(t)$ are the null and the non-null
densities, respectively; $p_0$ is the proportion of the null group, and $p_1=1-p_0$. 
We define the null sub-density as  $p_0 f_0(t)$. % and the null sub-survival function as $p_0 S_0(t)$. 
The local false discovery rate (denoted by fdr) and the right tail FDR (denoted by Fdr) at
$t$ are, respectively, defined as
\begin{equation} \label{eqn:fdr_definition}
{\rm fdr}(t) = \frac{p_0 f_0(t)}{f(t)} \quad \mbox{and } \quad
{\rm Fdr}(t)=\frac{p_0 S_0(t)}{p_0 S_0(t) + p_1 S_1(t)},
\end{equation}
where $S_0(t)$ and $S_1(t)$ are the survival functions of the null and the non-null groups, respectively. Note that the tail FDR corresponds to one-sided hypotheses toward the positive side, and the other direction or the left tail counterpart can be similarly defined.

Knowledge of the null density $f_0(t)$ plays a crucial role in the inference regarding
fdr and Fdr. The null distribution of the test statistics for single hypothesis testing is often known theoretically, e.g., standard normal, Student's $t$, or chi-square.
However, in multiple hypothesis testing, the observed test statistics often do not follow the theoretical null distribution. This phenomenon may be due to failed assumptions, unobserved covariates, correlations among the samples or among the test statistics \citep{Efron:2007b}.

To remedy this problem, several authors advocate a family of empirical Bayes procedures, referred to as the empirical null method \citep{Efron:2007a,Efron:2007b,Schwartzman:2008}. 
This method estimates the null distribution from the data itself. For $N$ sufficiently large, the components of the mixture density \eqref{eqn:model} can be estimated under a certain set of assumptions. These assumptions include that $f_0(t)$ is unimodal, and that the most of the probability mass around the peak of $f(t)$ is due to the null sub-density $p_0 f_0(t)$.
Therefore, a reliable estimation of $f_0(t)$ and $p_0$ is very important for accurate inference of the FDRs discussed above.

To estimate  $f(t)$, $f_0(t)$, and $p_0$, \citet{Efron:2007b} proposes 
two methods, named  ``central matching" and ``MLE fitting". First, central matching is a two-step procedure. At step 1, the mixture density $f(t)$ is modeled as a semi-parametric exponential family, e.g.,
$
f(t) = c_{\beta} \exp \big\{ \sum_{j=1}^7 \beta_j t^j
\big\},
$
where $c_{\beta}$ is a normalization constant.
Subsequently the $N$ test statistics are binned into $K$ bins with equal width $\Delta$ centered at $t_1,t_2,\ldots,t_K $.
Let $y_k$ be the count in bin $k$. Then the parameters $\{\beta_j\}$ are fitted to $\{y_k\}$ using Lindsey's method \citep{Lindsey:1974}. At step 2, $f_0(t)$ is fit to the estimated $f(t)$ around $t=0$. Assuming $f_0(t)$ is a normal density, the parameters (mean and variance) for $f_0(t)$ are estimated by least squares. 
Second, MLE fitting  undertakes  maximum likelihood estimation, in which 
it is assumed that the non-null density is only supported outside some known interval $[t_{\min}, t_{\max}]$, i.e., $f_1(t)=0$ for $t \in [t_{\min}, t_{\max}]$, and the null density is  normal with unknown mean and variance.
The likelihood function of the $N$ test statistics is a product of a binomial and a truncated normal likelihoods. Then the parameters, i.e., $p_0$ and the mean and the variance of the null, are estimated by maximizing the product likelihood.

Central matching has been further generalized with general exponential families by \citet{Schwartzman:2008} (``mode matching'').  Assuming that the null density is taken from an exponential family
$
f_0(t)=a_0(t) \exp ( \mathbf{x}(t)^T \mathbf{\eta} - \psi(\mathbf{\eta}) ),
$
and the counts ${y_k}$ within $[t_{\min}, t_{\max}]$ are independent Poisson variables with mean $\lambda_k \approx N \Delta p_0 f_0(t_k)$,
the following Poisson regression model is obtained:
\begin{equation}\label{eqn:poisson}
\log( {\mathbf \lambda}) = \bf{X} \bf{\eta}^{+} + \bf{h},
\end{equation}
where $\mathbf{\lambda} = (\lambda_1, \ldots, \lambda_K)^T$; $\mathbf{\eta}^{+} = (C, \mathbf{\eta})^T$ with $C=\log p_0 - \psi(\mathbf{\eta})$; $\mathbf{X}$ is the design matrix with rows $(1, \mathbf{x}(t_k)^T)$, $k=1, \ldots, K$; and $\mathbf{h}=(h_1, \ldots, h_K)$ is a known offset vector with $h_k = \log(N\Delta a_0(t_k))$. Solving \eqref{eqn:poisson} provides an estimate vector $(\hat{C}, \hat{\mathbf{\eta}})^T$, from which the proportion of the null group  $\hat{p_0} = \exp( \hat{C} + \psi(\hat{\mathbf{\eta}}) )$ is reconstructed.
Then the estimates of the fdr and the Fdr at the bin centers $\{t_k\}$  are evaluated by
\begin{equation} \nonumber
\widehat{\rm fdr}(t_k) = \frac{p_0
\widehat{f}_0(t_k)}{\widehat{f}(t_k)}= \frac{\widehat{y}_k}{y_k} \quad \mbox{and} \quad
\widehat{\rm Fdr} (t_k) = \frac{ (1/2) \widehat{y}_k +
\sum_{j=k+1}^K \widehat{y}_j} { (1/2) y_k + \sum_{j=k+1}^k y_j},
\end{equation}
where $(\hat{y_1}, \ldots, \hat{y_K})^T = \hat{\mathbf{y}} = \exp( \mathbf{X}\hat{\mathbf{\eta}}^{+} + \mathbf{h} )$
is the vector of the expected frequencies of the bins.
Equivalently, in vector form, they are written as
\begin{equation}\label{eqn:logfdr}
\log \widehat{\bf fdr} = \log \widehat{\bf y}- \log {\bf y}
\quad \mbox{and } \log \widehat{\bf Fdr} = \log \big( {\bf S}
\widehat{\bf y} \big) - \log \big( {\bf S} {\bf y} \big),
\end{equation}
where $\mathbf{S}$ is an upper triangular matrix with entries $1/2$
on the diagonal and $1$ above the diagonal.

\section{Monotone False Discovery Rate}\label{sec:monotone}
In this section, we first examine attractive features of the monotone FDR (local FDR and tail FDR). We then 
propose a procedure to monotonize the estimates of $\mathrm{fdr}(t)$ and $\mathrm{Fdr}(t)$, and adaptive 
optimal procedure using the monotonized estimates. 
In the remainder of the section, we assume that $\mathrm{fdr}(t)$ is monotonically decreasing.

\subsection{Properties of monotone FDR}\label{sec:monotone:props}

\paragraph{Monotone local FDR is equivalent to monotone likelihood ratio} 
 Recall that
\begin{equation} \nonumber
{\rm fdr} (t) = \frac{ p_0 f_0(t)}{p_0 f_0(t) + p_1 f_1 (t)} =
\frac{p_0}{p_0 + p_1 f_1(t) \big/ f_0(t)}.
\end{equation}
This shows that monotone decrease of ${\rm fdr}(t)$ is equivalent to monotone increase of
the likelihood ratio $f_1(t) \big/ f_0(t)$. 
This equivalence in turn defines a stochastic ordering between the null and the alternative densities: the alternative density $f_1$ is said to be stochastically larger than the null density $f_0$ if the likelihood ratio is monotonically increasing \citep{RobertsonWD:1988, Lim:2012}. A similar statement can be made for the tail FDR. 

\paragraph{Monotone local FDR and the MLRC} For a random variable $T$ that has the identical distribution to the common marginal distribution of $T_1, \ldots, T_N$,
write the marginal density of $\Phi=\mathrm{fdr}(T)$ as $p_0 g_0(\phi) + p_1 g_1(\phi)$. Here $g_0$ and $g_1$ are the conditional densities of $\Phi$ under the null and the non-null, respectively. Since $\mathrm{fdr}(t)$ is monotone decreasing in $t$, $\mathrm{fdr}^{-1}(\cdot)$ is well-defined and 
\begin{align}\label{eq:smlr}
\frac{g_1(\phi)}{g_0(\phi)} = \frac{f_1(\mathrm{fdr}^{-1}(\phi))/(\mathrm{fdr}^{-1}(\phi))'}{f_0(\mathrm{fdr}^{-1}(\phi))/(\mathrm{fdr}^{-1}(\phi))'} = \frac{f_1(\mathrm{fdr}^{-1}(\phi))}{f_0(\mathrm{fdr}^{-1}(\phi))}
\end{align}
is decreasing in $\phi$.
Hence the oracle statistic $\mathrm{fdr}(T)$ has monotone likelihood ratio. This is precisely the MLRC of \citet{SunCai:2007}.
Note that if a statistic $\mathcal{T}(T)$ satisfies the MLRC, the decision rule $I\{\mathcal{T}(T) < c\}$ has many attractive features for multiple testing problems:
\begin{proposition}\label{prop:smlr}
\citep[Proposition 1]{SunCai:2007}. When $N$ summarizing statistics $T_1, T_2, \ldots, T_N$ follow the two-group mixture model \eqref{eqn:model}, if a statistic $\mathcal{T}(T_i)$ satisfies the MLRC, then applying the decision rule $I\{\mathcal{T}(T_i) < c\}$ for $i=1, \ldots, N$ implies
    \begin{enumerate}
    \item $\mathbf{Pr}(\text{non-null } | \mathcal{T}(T_i) < c)$ is monotonically decreasing in threshold $c$,
    \item mFDR is monotonically increasing in $c$ and the expected number of rejections $r$, and
    \item  mFNR is monotonically decreasing in $c$, $r$, and mFDR,
    \end{enumerate}
where $\rm{mFDR}$ is the marginal false discovery rate, or $\mathbf{Pr}(\mathcal{T}(T_i) < c,\text{null })/\mathbf{Pr}(\mathcal{T}(T_i)<c)$, and $\rm{mFNR}$ is the marginal false non-discovery rate, or $\mathbf{Pr}(\mathcal{T}(T_i)>c,\text{non-null })/\mathbf{Pr}(\mathcal{T}(T_i)>c)$.
\end{proposition}

\paragraph{Optimality of the monotone local FDR}
Not only that it has many good properties as a statistic for multiple testing, the monotone local FDR is optimal:
\begin{theorem}
If the local FDR in \eqref{eqn:fdr_definition} is monotonically decreasing, then for any given mFDR level $\alpha$ in a multiple testing problem on the summarizing statistics $T_1, T_2, \ldots, T_N$, there exists a unique $c(\alpha)$ such that the decision rule $I\{\mathrm{fdr}(T_i) < c(\alpha) \}$ has an mFDR not greater than $\alpha$ and the smallest mFNR among all decision rules of the form $I\{\mathcal{T}(T_i) < c\}$, where $\mathcal{T}$ satisfies the MLRC and $c$ can be any constant.
\end{theorem}
\begin{proof}
This can be easily proved by using Theorems 1 and 2 of \citet{SunCai:2007}, and the above result that the MLRC holds for the monotone local FDR.
\end{proof}

\paragraph{Monotonicity of the local, tail, and marginal FDRs}

If $\mathrm{fdr}(t)$ is monotonically decreasing, then the mFDR of the decision rule $I\{\mathrm{fdr}(T)< c\}$ is written as
\[
\frac{\mathbf{Pr}(\mathrm{fdr}(T)<c,~\text{null})}{\mathbf{Pr}(\mathrm{fdr}(T)<c)} = \frac{p_0 S_0(\mathrm{fdr}^{-1}(c))}{p_0 S_0(\mathrm{fdr}^{-1}(c)) + p_1 S_1(\mathrm{fdr}^{-1}(c)) } = \mathrm{Fdr}(\mathrm{fdr}^{-1}(c)).
\]
By Proposition \ref{prop:smlr}, mFDR is monotonically increasing in $c$. Hence $\mathrm{Fdr}(t)$ is monotonically decreasing in $t=\mathrm{fdr}^{-1}(c)$. Furthermore, the tail FDR can be controlled by controlling the local FDR:
\begin{proposition}
Assume $f_0(t)$ and $f_1(t)$ are continuous and positive for every $t$.
If $\mathrm{fdr}(t)$ is monotonically decreasing in $t \in \mathbb{R}_+$, then for every $\alpha \in(0,1)$, we have
\[
\big\{t: \mathrm{fdr}(t)\le \alpha \big\} \subset \big\{t: \mathrm{Fdr}(t) \le \alpha \big\}.
\]
\end{proposition}
\begin{proof}
Let $t$ and $\alpha$ be arbitrary numbers between 0 and 1, and let $t_{\alpha}$ be the unique root of the equation $\mathrm{fdr}(t)=\alpha$. Then, for $t \ge t_{\alpha}$, 
$
{f_1(t)}\big/{f_0(t)} \le {f_1(t_{\alpha})}\big/{f_0(t_{\alpha})} = ( {p_0}/{\alpha}-p_0 )\big/{p_1}
$.
Now the definition of ${\rm fdr}(t)$ and a simple algebra show that  
\begin{align*} 
\frac{S_1 (t_{\alpha})}{S_0(t_{\alpha})} &=  \frac{1}{S_0 (t_{\alpha})} \int_{t_{\alpha}}^{\infty} f_0(s)  \big( f_1(s) \big/ f_0(s) \big) ds \\
&\ge \frac{1}{S_0(t_{\alpha})}  \int_{t_{\alpha}}^{\infty} f_0(s) ds \cdot \bigg( \frac{p_0}{\alpha}-p_0 \bigg) \frac{1}{p_1} = \bigg( \frac{p_0}{\alpha} - p_0 \bigg) \frac{1}{p_1},
\end{align*}
which tells 
\begin{equation} \nonumber
{\rm Fdr} (t_{\alpha}) = \frac{p_0}{p_0 + p_1 S_1(t_{\alpha}) \big/ S_0 (t_{\alpha}) } \le \alpha.
\end{equation}
%\begin{eqnarray}
%\frac{S_1 (t_{\alpha})}{S_0(t_{\alpha})} &=& \frac{1}{S_0 (t_{\alpha})} %\int_{t_{\alpha}}^{\infty} f_1(s) ds 
%=  \frac{1}{S_0 (t_{\alpha})} \int_{t_{\alpha}}^{\infty} f_0(s)  \big( f_1%(s) \big/ f_0(s) \big) ds \nonumber \\
%&\ge& \frac{1}{S_0(t_{\alpha})}  \int_{t_{\alpha}}^{\infty} f_0(s) ds \cdot \bigg( \frac{p_0}{\alpha}-p_0 \bigg) \frac{1}{p_1} = \bigg( \frac{p_0}{\alpha} - p_0 \bigg) \frac{1}{p_1}, \nonumber
%\end{eqnarray}
%where, for $t \ge t_{\alpha}$,  
%\begin{equation} \nonumber
%\frac{f_1(t)}{f_0(t)} \le \frac{f_1(t_{\alpha})}{f_0(t_{\alpha})} = \bigg( \frac{p_0}{\alpha}-p_0 \bigg) \frac{1}{p_1}.
%\end{equation}
%Therefore,
%\begin{equation} \nonumber
%{\rm Fdr} (t_{\alpha}) = \frac{p_0}{p_0 + p_1 S_1(t_{\alpha}) \big/ S_0 %(t_{\alpha}) } \le \alpha,
%\end{equation}
%proves the statement of the theorem. 
\end{proof}

\paragraph{Estimated local FDR statistic}
As a final note, recall that $\mathrm{fdr}(T)$, is not a \textit{bona fide} statistic unless $f_0$, $f_1$, and $p_0$ are known \textit{a priori}. As a solely data-driven, hence \textit{bona fide}, statistic, we may consider an estimator $\widehat{\mathrm{fdr}}(T)$ of $\mathrm{fdr}(T)$, using the methods in Section \ref{sec:EB}. However, the resulting finite-sample estimator of the local FDR is not necessarily monotone, even if the true local FDR is. Hence it is desirable to incorporate monotonicity in the estimation procedure. 

Post-hoc monotonization of the local FDR estimates in the next section is attractive in the following sense. If the true local FDR is monotone,  then the monotonized local FDR satisfies the MLRC and yields a decision rule that enjoys the good properties listed in the beginning of the section. This is readily seen by plugging in a monotone estimator $\widehat{\mathrm{fdr}}(\cdot)$ in place of the true (monotone) $\mathrm{fdr}(\cdot)$ in \eqref{eq:smlr}. Furthermore, if the unadjusted estimate is consistent, then monotonization preserves consistency while reducing variance.

\subsection{Estimation by monotonization}\label{sec:monotone:estimation}

We propose to modify the FDR estimates by imposing a monotone ordering (``isotonization'') among them. Suppose the mode matching method by \citet{Schwartzman:2008} is employed to estimate the local FDR and the tail FDR. Using the delta method, the  variance-covariance matrices of  
$\log \widehat{\bf fdr}$ and $\log \widehat{\bf Fdr}$ are computed as follows. 
Let ${\bf X}$ be the design matrix in \eqref{eqn:poisson},
and ${\bf W}$ be the diagonal matrix made of the vector ${\bf
w}=\big(w_1,w_2,\ldots,w_K \big)$, where $w_k$ is equal to 1 or 0
according to whether $t_k$ is in the null region $[t_{\min}, t_{\max}]$ for the Poisson
regression.  Set $\widehat{\bf V}= {\rm diag} \big(\widehat{\bf y} \big)$,  
$\widehat{\bf V}_{N}=\widehat{\bf V}-\widehat{\bf y} \widehat{\bf y}^{T} \big/ N$, and
${\bf D}_y = {\bf X} \big( {\bf X}^{T} {\bf W} \widehat{\bf V} {\bf X} \big)^{-1} {\bf X}^{T} {\bf W}$.
Then the desired variance-covariance matrices are given as
\begin{equation} \nonumber
\widehat{\rm cov} \big( \log \widehat{\bf fdr} \big) = {\bf A}
\widehat{\bf V}_{\rm N} {\bf A}^{T}
\quad\mbox{and}\quad
\widehat{\rm cov} \big( \log \widehat{\bf Fdr}\big) = {\bf B}
\widehat{\bf V}_{\rm N} {\bf B}^{T},
\end{equation}
where ${\bf A} = {\bf D}_y - {\bf V}^{-1}$ and $ 
{\bf B}= \widehat{\bf U}^{-1} {\bf S}
\widehat{\bf V}^{-1} {\bf D}_y - {\bf U}^{-1}$
with ${\bf U}={\rm diag} \big( {\bf S} {\bf y} \big)$ and
$\widehat{\bf U} = {\rm diag} \big({\bf S} \widehat{\bf y} \big)$ \citep{Schwartzman:2008}.

We proceed to adjust the initial estimate $\widehat{\bf fdr}$ by solving the quadratic prog):
\begin{equation} \label{eqn:fdr-opt1}
\begin{array}{ll}
\mbox{minimize} & \big( {\bf z} - \log \widehat{\bf fdr}
\big)^{T} \widehat{\rm cov} \big( \log \widehat{\bf fdr}
\big)^{-1} \big( {\bf z} - \log \widehat{\bf fdr} \big) \\
\mbox{subject to} & z_1 \le z_2 \le \cdots \le z_K
\end{array}
\end{equation}
for $\mathbf{z} = (z_1, \ldots, z_K)^T$. 
This QP is a convex optimization problem that can be efficiently solved using existing software packages, e.g., \texttt{quadprog} R package.

If $K$ is large, we suggest to solve a simplified version of \eqref{eqn:fdr-opt1}:
\begin{equation} \label{eqn:fdr-opt2}
\begin{array}{ll}
\mbox{minimize} & \big( {\bf z} - \log \widehat{\bf fdr}
\big)^{T} {\rm diag} \Big\{ \widehat{\rm cov} \big( \log
\widehat{\bf fdr}
\big)\Big\}^{-1} \big( {\bf z} - \log \widehat{\bf fdr} \big) \\
\mbox{subject to} & z_1 \le z_2 \le \cdots \le z_K.
\end{array}
\end{equation}
This is a generalized isotonic
regression problem and can be solved using the pool-adjacent-violator (PAVA) algorithm of  \citet{RobertsonWD:1988}. A similar procedure can be applied to monotonize $\widehat{\bf Fdr}$, estimates of Fdr.

\subsection{Adaptive Decision Rule with Monotonized Estimate}\label{sec:monotone:decisionrule}

Suppose we have obtained monotonized estimates $\widehat{\bf{fdr}}^{\mathrm{iso}}(t_k)$ and $\widehat{\bf{Fdr}}^{\mathrm{iso}}(t_k)$ of the local and the tail FDRs  at $t=t_k$, respectively. Let $\widehat{\bf{fdr}}^{\mathrm{iso}}_{(k)}$ and $\mathcal{H}^{(k)}$ be the $k$th largest value and its corresponding null hypothesis. 
Following \citet{SunCai:2007},
we propose a decision rule that is step-up and rejects all hypotheses $\mathcal{H}^{(k)}$, $k=1,2,\ldots, u$, where
\begin{align}\label{eqn:decisionrule}
u=\max \Big\{j  ~\big|~(1/j)  \sum_{k=1}^j \widehat{\bf{fdr}}^{\mathrm{iso}}_{(k)} \le \alpha \Big\} \quad \mbox{and} \quad 
u=\max \Big\{j ~ \big|~ \widehat{\bf{Fdr}}^{\mathrm{iso}}_{(j)} \le \alpha \Big\}
\end{align}
for the local FDR and the tail FDR, respectively.

The above decision rule \eqref{eqn:decisionrule} suggests that it suffices to monotonize the tail region of the initial estimates $\widehat{\bf fdr}$ and $\widehat{\bf Fdr}$. It seems reasonable to monotonize these estimates outside of $[t_{\min}, t_{\max}]$ where $f_1(t)=0$ is assumed.

\section{Numerical Study}\label{sec:numerical}

In this section, we compare the performance of the monotonized FDR estimators in Section \ref{sec:monotone:estimation} to the unadjusted estimators numerically. 
The same numerical scheme as in \citet[Section 4]{Schwartzman:2008} is used for this study.

Consider two scenarios to generate sets of summarizing statistics from the two-group mixture model \eqref{eqn:model}: 
$T_i$s are independent random variables from the following mixture models
\begin{equation}  \nonumber
T_i \sim \left\{
\begin{array}{ll}
f_0 = N(0.2,1.2^2)&~ \text{w.~p.~~} p_0,  \\
%& \\
f_1= N(3,1.2^2)&~ \text{w.~p.~~}p_1,
\end{array}
\right. \quad \text{or}\quad
 T_i \sim \left\{
\begin{array}{ll}
f_0 = 0.8 \chi^2(3) &~ \text{w.~p.~~} p_0,\\
%& \\
f_1= \chi^2(3,3) &~\text{w.~p.~~}p_1,
\end{array}
\right.
\end{equation}
where $N(\mu,\sigma^2)$ denotes the normal density with mean $\mu$ and variance $\sigma^2$,  $a \chi^2(\nu)$ denotes the scaled chi-square distribution with $\nu$ degrees of freedom, 
and $\chi^2(\nu,\delta)$ denotes the non-central chi-square distribution with the non-centrality parameter $\delta$.
In the study, we assume $p_0=0.9$ and generate 100 data sets for each case. 
The fitting interval for estimating the
null density is set to be $[0.2-1.5,0.2+1.5]$ for the normal
case and $[0,4]$ for the chi-square case. In both cases, only the right tail is monotonized
using the simpler procedure \eqref{eqn:fdr-opt2}.
% The fitting interval to estimate the
% null distribution is set to be $[0.2-t_0,0.2+t_0]$ for the normal
% case and $[0,t_0]$ for the $\chi^2$ case. 
% {\color{red}
% We set $t_0=1.5$ for the normal scenario and $t_0=4.0$ for the chi-square scenario; the bin width $\Delta$ is fixed to 0.1.
% }
% We fix the bin width $\Delta=0.1$ and $t_0=1$ and apply the mode matching method and
% our proposal to estimating ${\rm fdr}(t)$ and ${\rm Fdr}(t)$, with and without the monotonicity constraint.  \textcolor{red}{In the study, we use simpler monotonization (\ref{eqn:fdr-opt2}) and the monotonization is applied to the tail estimates of ${\rm fdr}(t)$ or ${\rm Fdr}(t)$ whose $t \ge c$ for some constant $c$. The constant $c$ is chosen as $2.0$ and $4.0$ for normal 
% and chi-square statistics, respectively.}

Figure \ref{fig:FDRplot} plots the \emph{average} of 100 fdr estimates and its $95\%$
``validity ranges'' of each method. The validity ranges 
are pointwise and computed with 2.5\% and 97.5\% quantiles of 100 fdr estimates at each $t$.
Figure \ref{fig:FDRplot} \subref{fig:FDRplot:normalfdr1} and \subref{fig:FDRplot:chisqfdr1} indicate
that the monotonized fdr estimates have smaller variance than their
unadjusted counterparts. In particular, the unadjusted
estimates for the chi-square case are quite volatile 
for large $t$ values, even after taking an average of the 100 estimates.
This volatility is substantially reduced after monotonization.
Furthermore, the validity range for the monotonized estimates is much
narrower than the unadjusted one.

Note that Figure \ref{fig:FDRplot} \subref{fig:FDRplot:normalfdr1} and \subref{fig:FDRplot:chisqfdr1} 
also indicate that the proposed monotonization reduces the the bias.
This phenomenon is interesting because smoothing as imposed by the monotonization procedure 
does not necessarily reduces the bias of the estimate. 
Our conjecture is that this is mainly due to the inverse-variance weighting used in the isotonic regression \eqref{eqn:fdr-opt2}.
This procedure imposes relatively small weights to  
the bins with small numbers of observations, whose unadjusted estimates for these bins are likely to be biased upward.
% due to the $-\log{\bf y}$ term in \eqref{eqn:logfdr}. This conjecture also explains why the averages of the monotonized estimates are placed below the unadjusted ones in Figure \ref{fig:FDRplot} \subref{fig:FDRplot:chisqfdr1}.

The unadjusted Fdr estimates are on average quite smooth and monotone as compared to their fdr counterparts.
In Figures \ref{fig:FDRplot} \subref{fig:FDRplot:normalFdr2} and \subref{fig:FDRplot:chisqFdr2}  the averages of 100 unadjusted and monotonized Fdr estimates are almost equal. 
Paying attention to the individual data set, however, 4 out of 100 data
sets result in unadjusted Fdr estimates in the normal case; 17 out of 100 result in unadjusted estimates in the chi-square case. 
If we limit our attention to these cases that do have unadjusted Fdr estimates, we observe that the isotonization step improves accuracy (Figure \ref{fig:FDRnonmonotone}).

\begin{figure}[htb!] %\label{}
\centering
    \subfigure[Normal fdr]{
    \includegraphics[width=.45\linewidth]{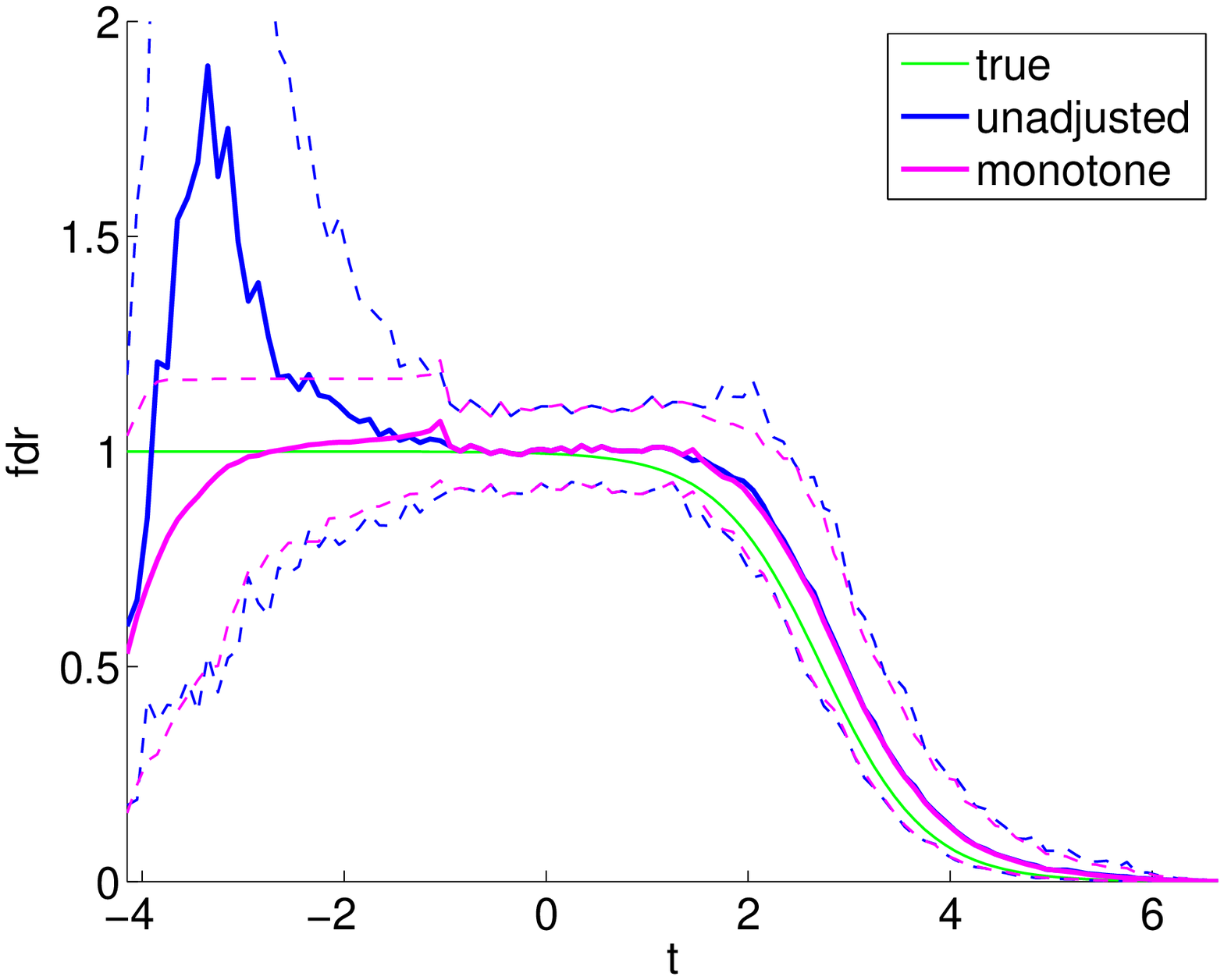}
    \label{fig:FDRplot:normalfdr1}
    }
    \subfigure[Normal Fdr]{
    \includegraphics[width=.45\linewidth]{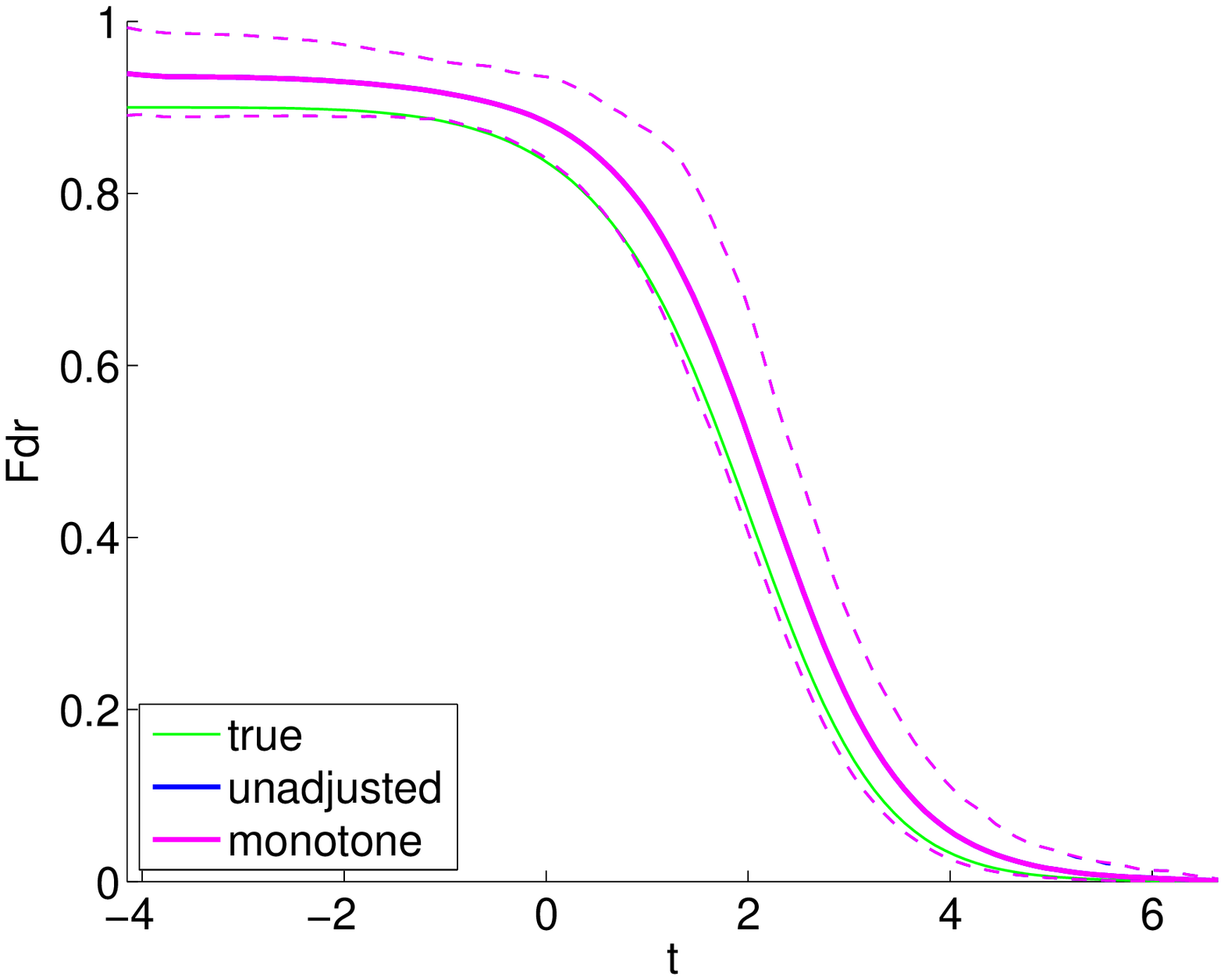}
    \label{fig:FDRplot:normalFdr2}
    }
    \subfigure[Chi-square fdr]{
    \includegraphics[width=.45\linewidth]{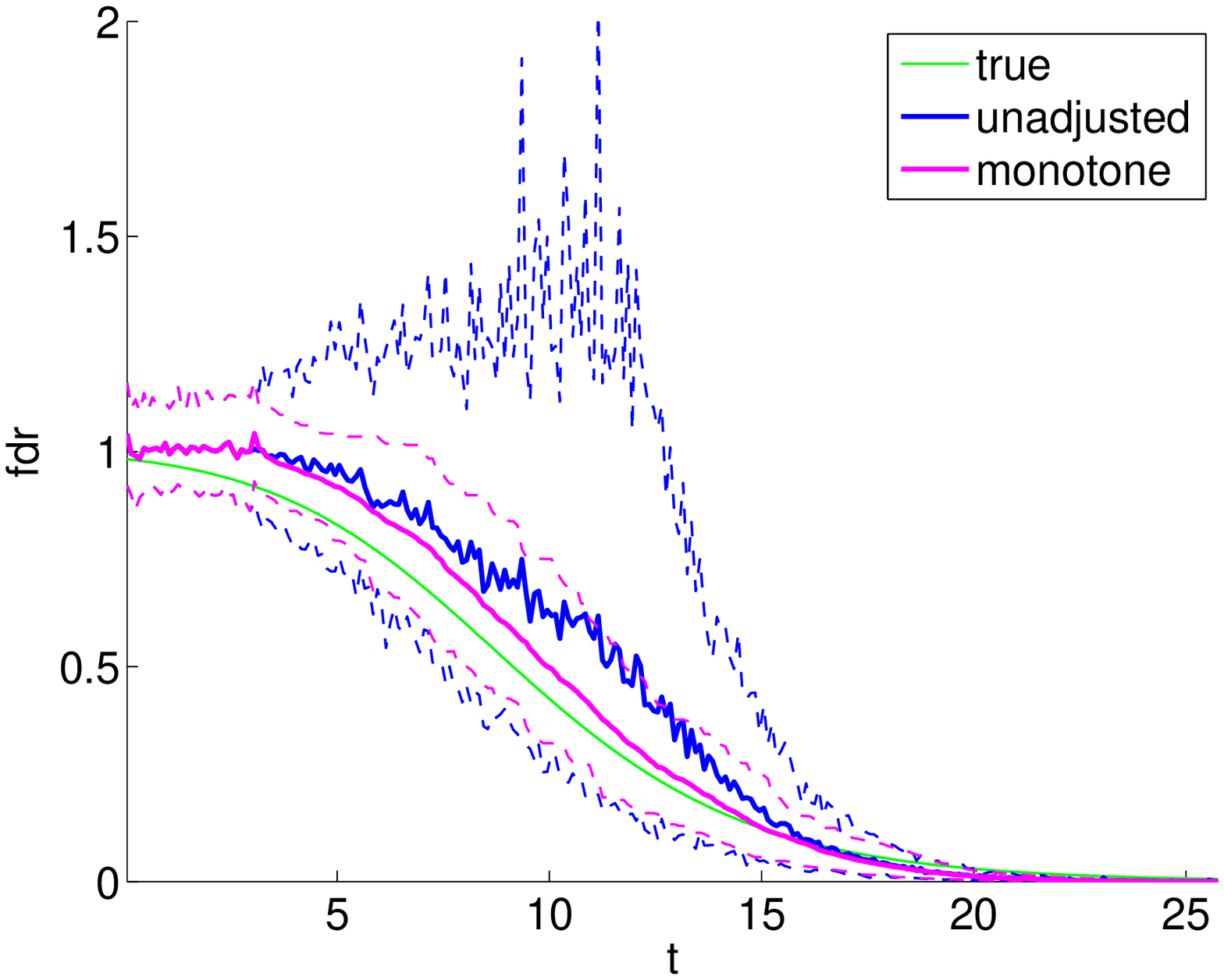}
    \label{fig:FDRplot:chisqfdr1}
    }
    \subfigure[Chi-square Fdr]{
    \includegraphics[width=.45\linewidth]{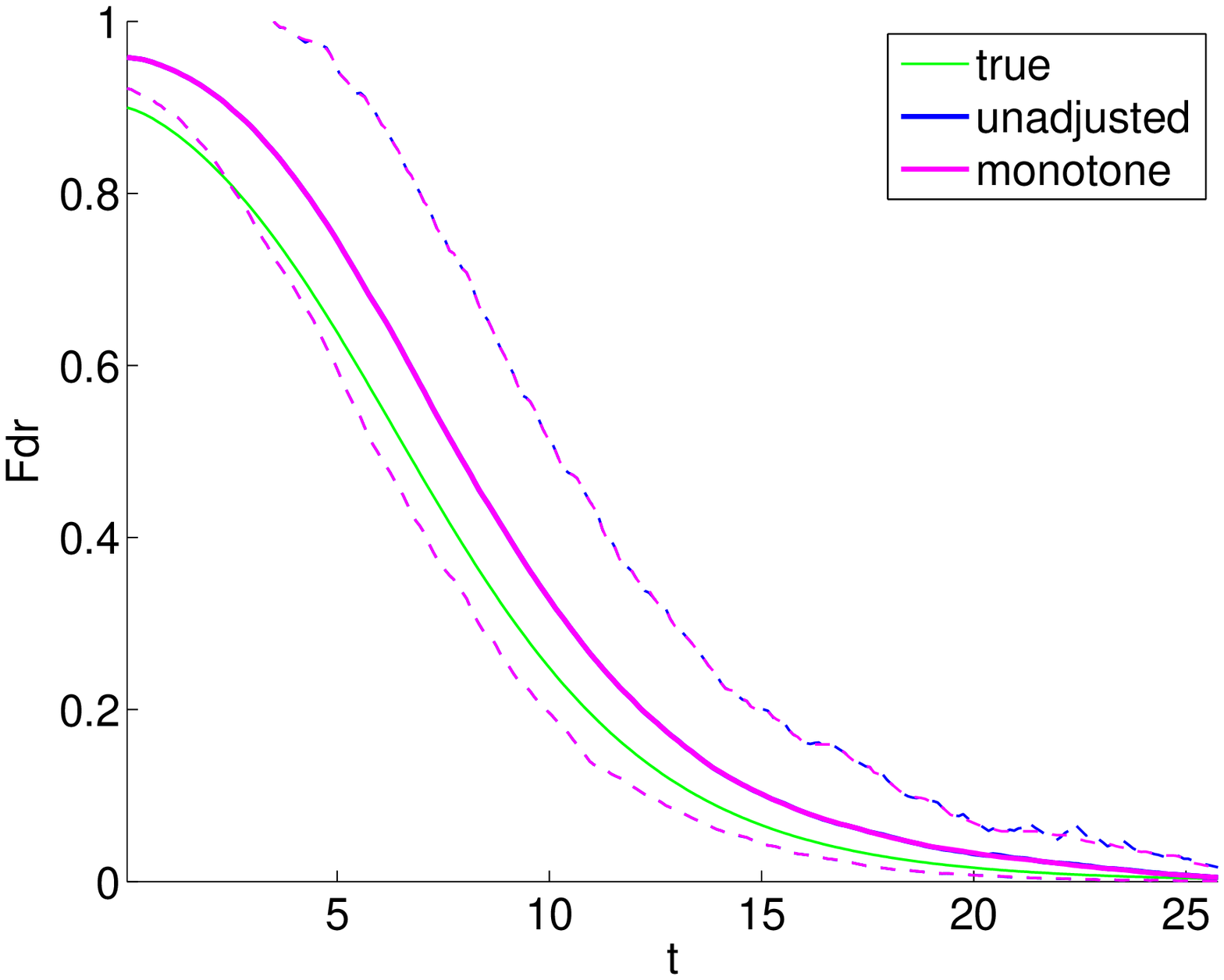}
    \label{fig:FDRplot:chisqFdr2}
    }
\caption{FDR bias and validity ranges. 
For panels (a) and (b), the unadjusted estimates are monotonized for $t>1.7$; for panels (c) and (d), $t>4.0$.
}
\label{fig:FDRplot}
\end{figure}

\begin{figure}[htb!] %\label{}
\centering
\includegraphics[width=.5\linewidth]{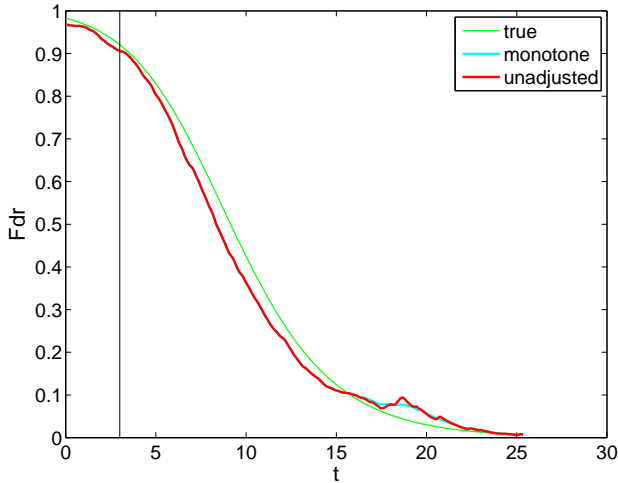}
\caption{
Plot of the averages of unadjusted Fdr estimates and their monotonized estimates for the chi-square scenario. 
The vertical line represent $t_0=4.0$, the boundary of monotonization.}
\label{fig:FDRnonmonotone}
\end{figure}

\section{Examples}\label{sec:examples}

In this section, we illustrate the merit of monotonization using the leukemia data by \citet{Golub:1999}, available from 
\url{http://www.broadinstitute.org/cancer/pub/all_aml/}.
This data set records the expression levels of patients with one of the two types of leukemia, acute lymphoblastic leukemia (ALL) or  acute myeloid leukemia (AML). The data set consists of two parts: training and test. 
The training data set is comprised of 38 arrays (ALL, 27; AML, 11). The test 
data set has 34 arrays (ALL, 20; AML, 14).  
In our analysis, we only used the training set, as in \citet{Broberg:2005}, to find differentially expressed genes (DEGs) between ALL and AML. Preprocessing was conducted according to the prescription due to \citet{Dudoit:2002}. The preprocessed data set was summarized as a $38 \times 3571$ matrix. (This date set is available from R package \texttt{multtest}).

We applied the mode matching procedure \citep{Schwartzman:2008} and monotonized the estimated fdrs and Fdrs using the method of Section \ref{sec:monotone:estimation} to find DEGs. 
For $g=1,2,\ldots,3571$, we computed two-sample $t$-statistics (with equal variance) $t_g$ as summarizing statistics, and transformed them to $z$-values $z_g= \Phi^{-1}\big(F_{36}(t_g) \big)$,
where $F_{36}(t)$ is the cumulative distribution function of the $t$-distribution with 36 degrees of freedom.
In applying the mode matching procedure, we chose the bin size $\Delta=0.05$ and the null region $\big[ -1.2, 1.2 \big]$ to estimate the empirical null distribution.  
We found the mode-matched estimates of the fdrs frequently violate the monotonicity, but the estimates of Fdrs did not need additional monotonization; fdrs are non-smooth and not monotone due to scarcity of observations in both tails. We only monotonized the fdrs outside the null region, i.e., those in the region $(-\infty, -1.2] \cup [1.2, \infty)$.  
Figure \ref{fig:estimation_leukemia} depicts the estimates of the local FDR and their monotonization.

\begin{figure}[htb!]
\centering
    \subfigure[Left tail]{
    \includegraphics[width=0.45\linewidth]{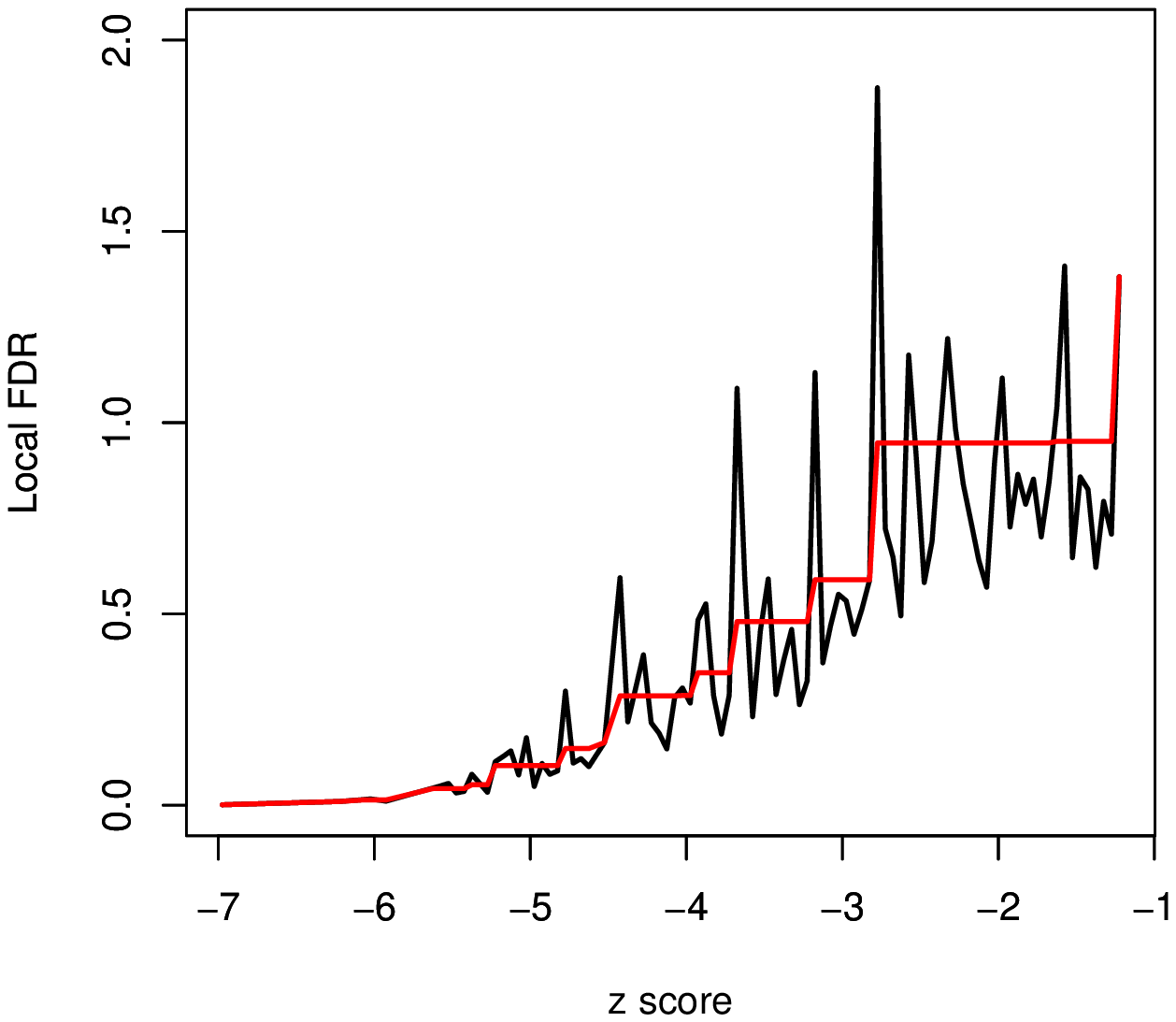}
    }
    \subfigure[Right tail]{
    \includegraphics[width=0.45\linewidth]{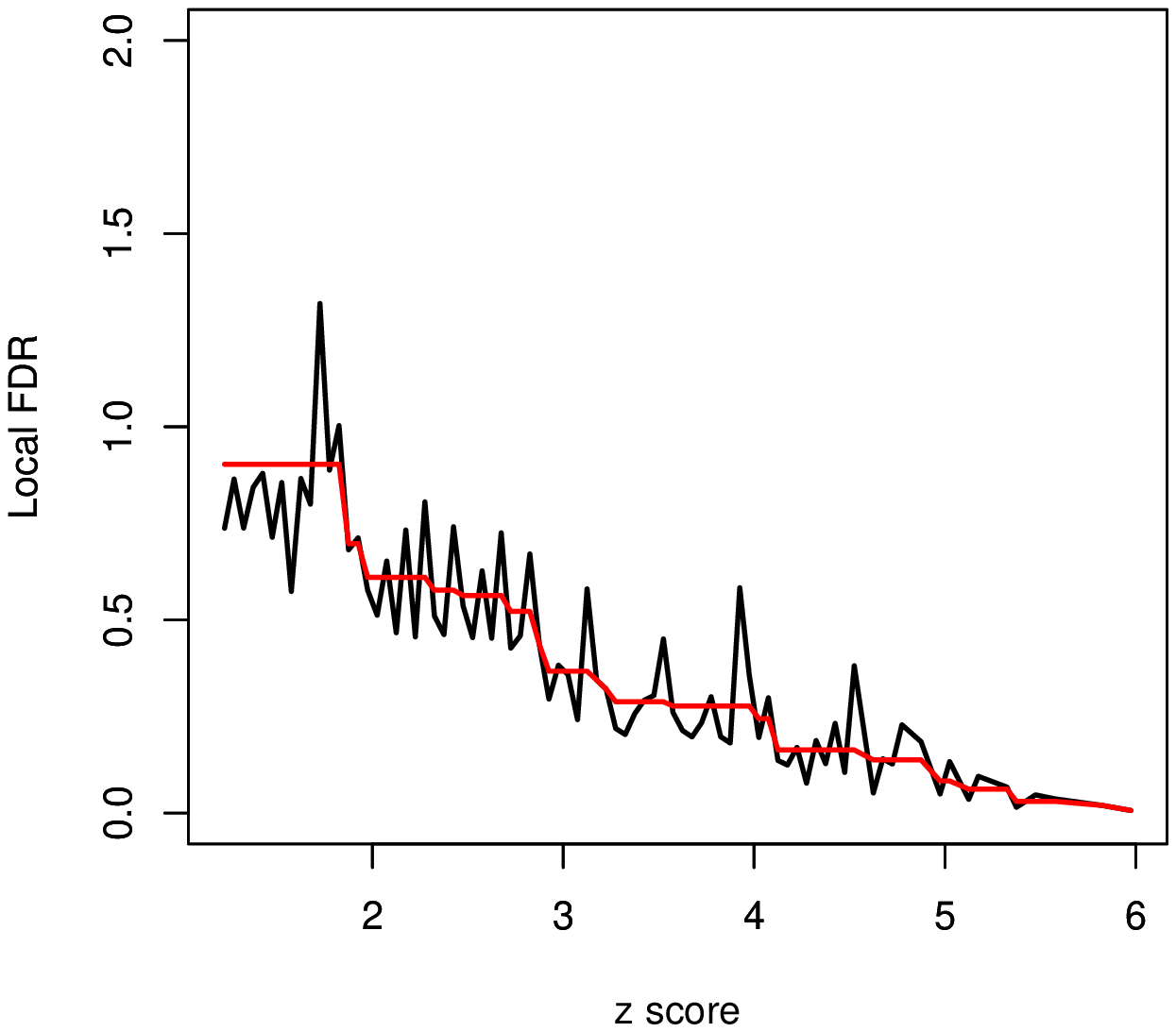}
    }
\label{fig:estimation_leukemia}
\caption{
Monotonized fdr estimates of leukemia data due to \citet{Golub:1999}. The spiky, black line represents the unadjusted estimates; the smooth, red line represents the monotonized local FDR estimates.}
\end{figure}

We then applied the adaptive decision rule of Section \ref{sec:monotone:decisionrule} to the unadjusted and the monotonized estimates of the local FDR to declare DEGs. We control the marginal FDR level at $\alpha=0.05, 0.1,$ and $0.15$. At level $\alpha=0.05$, we found 68 DEGs using the unadjusted fdr estimates (denoted by \texttt{unadj}) and 40 DEGs by using their monotonized modifications (denoted by \texttt{iso}).
Using HuGE Navigator version 2.0 databases \citep{Yu:2008}, we investigated biological relevance of the DEGs found by seeking AML/ALL-related
genes among those genes. Our data set has 2625 unique
genes (from the 3571 probes), among which the number of AML/ALL-related genes reported by HuGE navigator was 130
(4.4\% of the 2625 genes). Our monotonization removed 32 genes from the 68 DEGs that
\texttt{unadj} found, and introduced 4 new genes, one of which was AML/ALL-related; 
only 2 of the 32 DEGs removed were AML/ALL-related genes. 
In short, the percentage of AML/ALL-related genes in detected
DEGs increased from 8.82\% to 12.5\% by taking into account monotonicity in the local FDR. This observation indicates that the isotonization can reduce the number of false discoveries.
This observation is still valid for levels $\alpha=0.1$ and $0.15$, although the improvement due to isotonization becomes smaller as the level increases. 
These results are summarized in Table 1.

\begin{table}[htb!]
\caption{Biological relevance of detected DEGs}
\begin{center}
\begin{tabular}{|c|c|c|c|c|} \hline
$\alpha$ & FDR & \# of DEGs & \# of AML/ALL-related
& \% of AML/ALL-related \\ \hline
 \multirow{2}{*}{0.05} & fdr (\texttt{iso}) & 40 & 5 & 12.50\% \\
    					& fdr (\texttt{unadj}) & 68 & 6 & 8.82 \% \\ \hline 
	% &DEGs				& intersection 36 & \texttt{unadj} only 32 & iso only 4\\
	% &AML/ALL realted   & intersection 4 & \texttt{unadj} only 2 & iso only 1\\ \hline     
 \multirow{2}{*}{0.1} & fdr (\texttt{iso}) & 125 & 9 & 7.20\% \\ 
						& fdr (\texttt{unadj}) & 177 & 11 & 6.21\% \\ \hline
	% &DEGs					& intersection 117 & \texttt{unadj} only 60 & iso only 8\\
	% &AML/ALL realted   & intersection 9 & uniso only 2 & iso only 0\\ \hline    
 \multirow{2}{*}{0.15} & fdr (\texttt{iso}) & 232 & 14 & 6.03 \%\\  
						& fdr (\texttt{unadj}) & 362 & 21 & 5.80 \% \\ \hline
	%&DEGs				& intersection 223 & \texttt{unadj} only 139 & iso only 9\\
	% &AML/ALL realted  & intersection 13 & \texttt{unadj} only 8 & iso only 1  \\   \hline

\end{tabular}
\end{center}
\end{table}

As another example, we analyzed the adenocarcinoma data set in \cite{Notterman2001}.
The data consist of 18 subjects for each of which 6579 gene expressions in the adenocarcinoma and normal colon samples are obtained and paired.
The results are reported in the supplementary material, and again indicate that monotonization can reduce the number of false discoveries.

\section{Conclusion}\label{sec:conclusion}

We have considered monotonicity in the FDR and proposed an estimation procedure thereof. 
The proposed procedure is a simple modification of the empirical Bayes estimator using generalized isotonic regression. 
The presented numerical study shows that imposing monotonicity improves the estimation in both bias and variance. Through real-world data sets, it is demonstrated that the proposed monotone FDR procedure can reduce the number of false discoveries.

Monotonicity in the FDR has several attractive features: monotone local FDR implies optimality in controlling the tail FDR, and monotonized estimates perform better than their unadjusted counterparts in practice. The latter may be due to that imposing smoothness (via requiring monotonicity) improves estimation as in many non-parametric regression problems.

% \section*{Acknowledgment}
% We are grateful to the associate editor and the reviewers for constructive comments to improve the presentation. 

%We thank Jayeon Kim for data preparation at the early stage of this paper. 
%J. Lim’s research was supported by the National Research Foundation of
%Korea(NRF) grant funded by the Korea government(MSIP) (No. 2011-0029104). B.S. Kim's work was supported by the National Research Foundation of Korea (NRF) Grant Funded by the Korean Goverment (MSIP) (No. 2010-0022549).

%\section*{References}
\bibliographystyle{elsarticle-harv}

\end{document}

\maketitle
\section*{Adenocarcinoma Example}
In this document, we provide an additional data example to show an advantage of the proposed monotone estimates of the local FDR. We analyzed the adenocarcinoma data set in \cite{Notterman2001}, which is available from \url{http://genomics-pubs.princeton.edu/oncology}.
The data set consists of 18 arrays of adenocarcinoma samples and another 18 
arrays of their paired normal colon samples. Each array records expression levels of 
7,457 genes. In the array, expression levels of some genes were repeatedly measured and they 
are replaced with their averages. After this, 6,579 genes remained in each array.

As in the Example section (Section 5 of main paper), to control the marginal FDR, we apply the adaptive optimal procedure by \citet{SunCai:2007} with non-monotonized estimates and monotonized estimates of the local FDR. We controlled the marginal FDR at $0.01,  0.025$ and $0.05$.

We found 61 DEGs by the adaptive procedure with the non-monotonized fdr estimates (denoted by \texttt{unadj})
at 0.01 and 50 DEGs by using their monotonized modifications (denoted by \texttt{iso}) with the same procedure.
% Using HuGE Navigator version 2.0 databases by Yu et al. (2008) (\url{http://hugenavigator.net}), we investigated biological relevance by seeking adenocarcinoma neoplasms-related genes among the DEGs found by \texttt{unadj} and \texttt{iso}, respectively. 
Using HuGE Navigator version 2.0 databases \citep{Yu:2008}, we investigated the biological relevance of the DEGs found by seeking adenocarcinoma neoplasms-related genes among those genes.
Our data set has 6,597 unique
genes out of 7,457 observed genes, among which the number of adenocarcinoma neoplasms-related genes is 312 (4.73\% of 6,597 genes). Our monotonized estimates removed 11 genes from the 61 DEGs found by 
non-monotonized estimates of the local FDR.  All 11 removed genes are not in the list of adenocarcinoma neoplasms-related genes. In short, the percentage of adenocarcinoma neoplasms-related genes in detected
DEGs increased from 9.84\% to 12\% by taking into account the monotonicity of the local FDR. This shows again 
that the proposed isotonization can reduce the number of false discoveries. We made similar observation 
at level $0.025$ and $0.05$ as shown in Table \ref{tbl:adeno}. The estimates of local FDR are plotted in Figure \ref{fig:adeno}. 

\begin{table}[htb!]
\centering
\caption{Biological relevance of detected DEGs}
\label{tbl:adeno}
\begin{tabular}{c|c|c|c|c}
\hline
\multirow{2}*{$\alpha$} & \multirow{2}*{FDR} & \multirow{2}*{\# of DEGs} & \# of Adenocarcinoma
& \% of of Adenocarcinoma\\ 
 &   &  & neoplasms-related
& neoplasms-related \\ \hline

  \multirow{2}{*}{0.01} & fdr(\texttt{iso}) & 50 & 6 & 12\% \\
						& fdr(\texttt{unadj}) & 61 & 6 & 9.84\% \\ \hline
 \multirow{2}{*}{0.025} & fdr(\texttt{iso}) & 101 & 10 & 9.90\% \\
						& fdr(\texttt{unadj}) & 115 & 11 & 9.57\% \\ \hline 
 \multirow{2}{*}{0.05} & fdr(\texttt{iso}) & 182 & 16 & 8.79\% \\
						& fdr(\texttt{unadj}) & 215 & 18 & 8.37\% \\ \hline 
 %\multirow{2}{*}{0.1} & fdr(\texttt{iso}) & 339 & 22 & 6.49\% \\ 
%						& fdr(\texttt{unadj}) & 397 & 30 & 7.56\% \\ \hline
% \multirow{2}{*}{0.15} & fdr(\texttt{iso}) & 461 & 34 & 7.38\%\\  
%						& fdr(\texttt{unadj}) & 591 & 42 & 7.11\% \\ \hline
\end{tabular}
\end{table}

\begin{figure}[htb!] 
\begin{center}
  \begin{minipage}[b]{.48\linewidth}
  \centering   \centerline{\epsfig{file=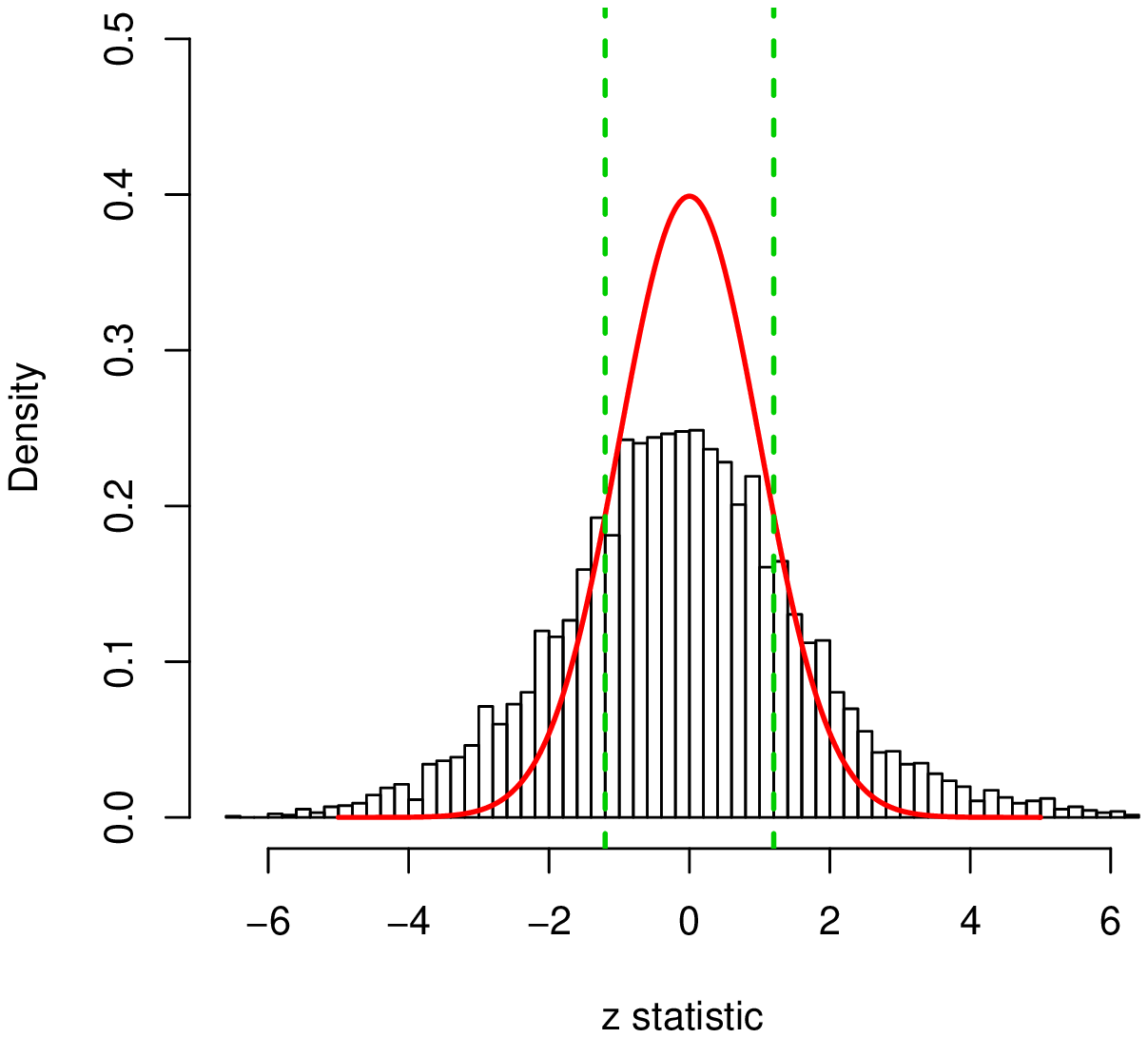,width=0.7\textwidth,height=0.25\textheight}}
  \centerline{(a) Transformed z-statistics}
\end{minipage}
\medskip
  \begin{minipage}[b]{.48\linewidth}
  \centering   \centerline{\epsfig{file=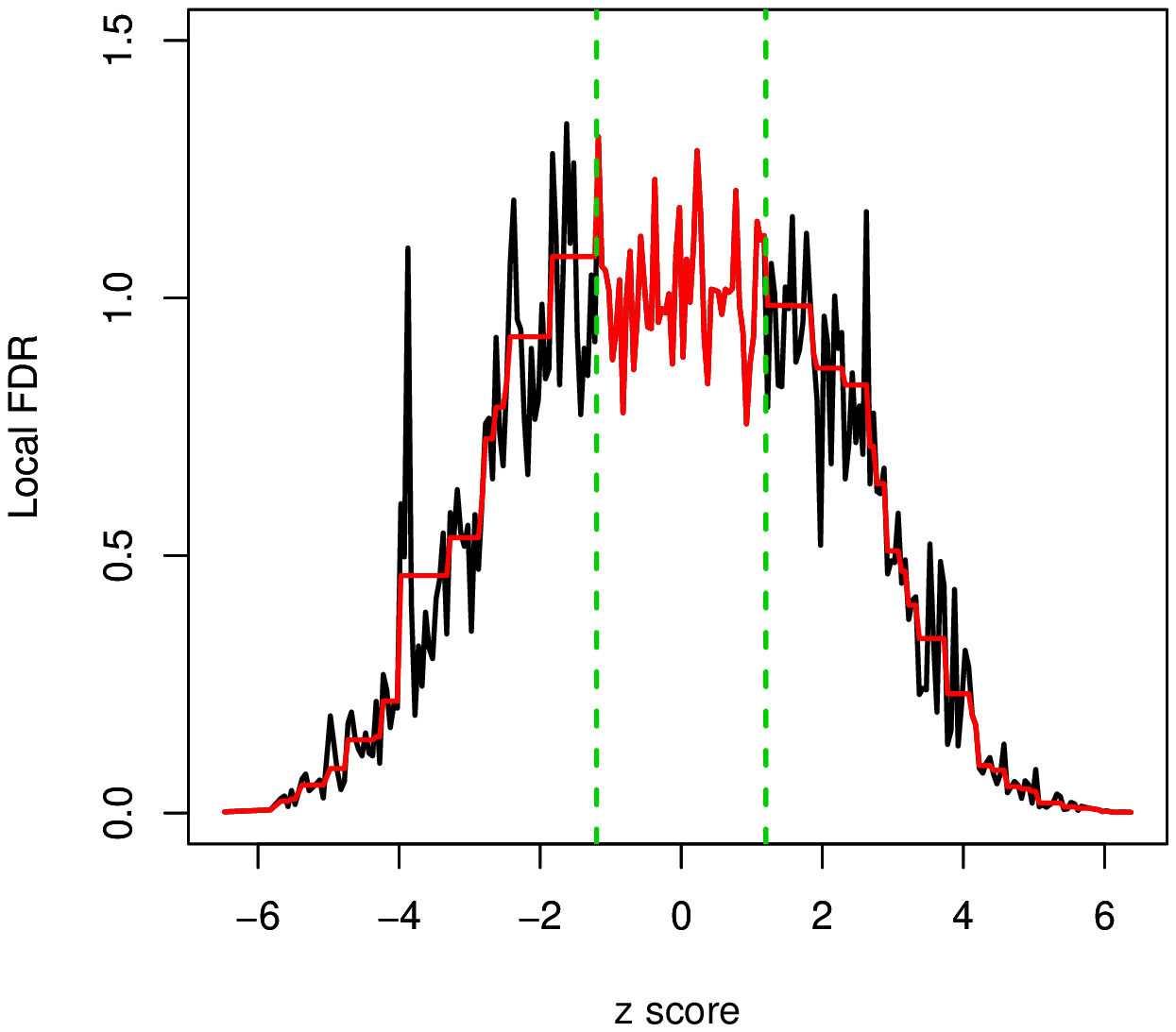,width=0.7\textwidth,height=0.25\textheight}}
  \centerline{(b) Estimated local FDR}
\end{minipage}
\medskip
  \begin{minipage}[b]{.48\linewidth}
  \centering   \centerline{\epsfig{file=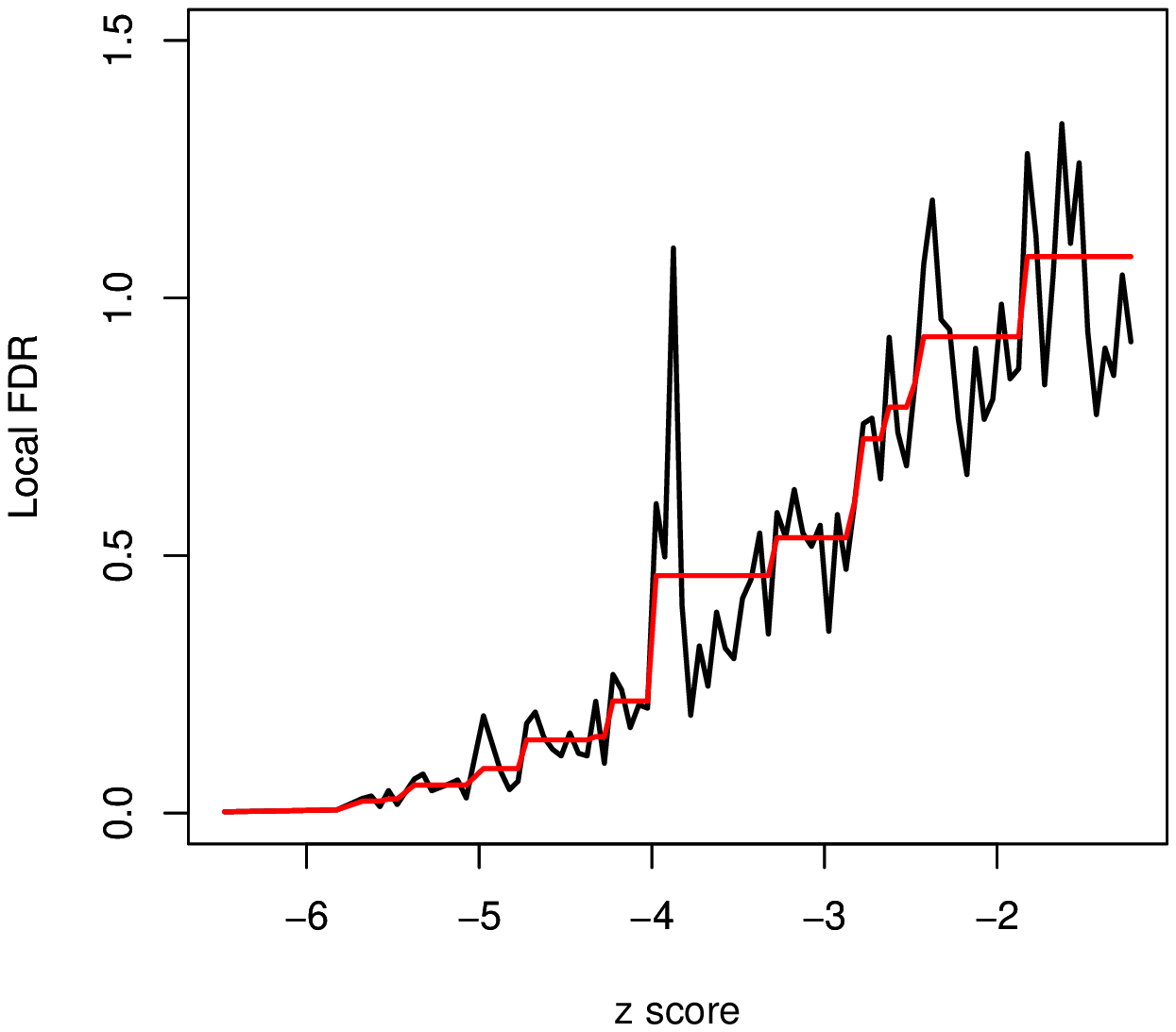,width=0.7\textwidth,height=0.25\textheight}}
  \centerline{(c) Estimated local FDR - left}
\end{minipage}
\medskip
  \begin{minipage}[b]{.48\linewidth}
  \centering   \centerline{\epsfig{file=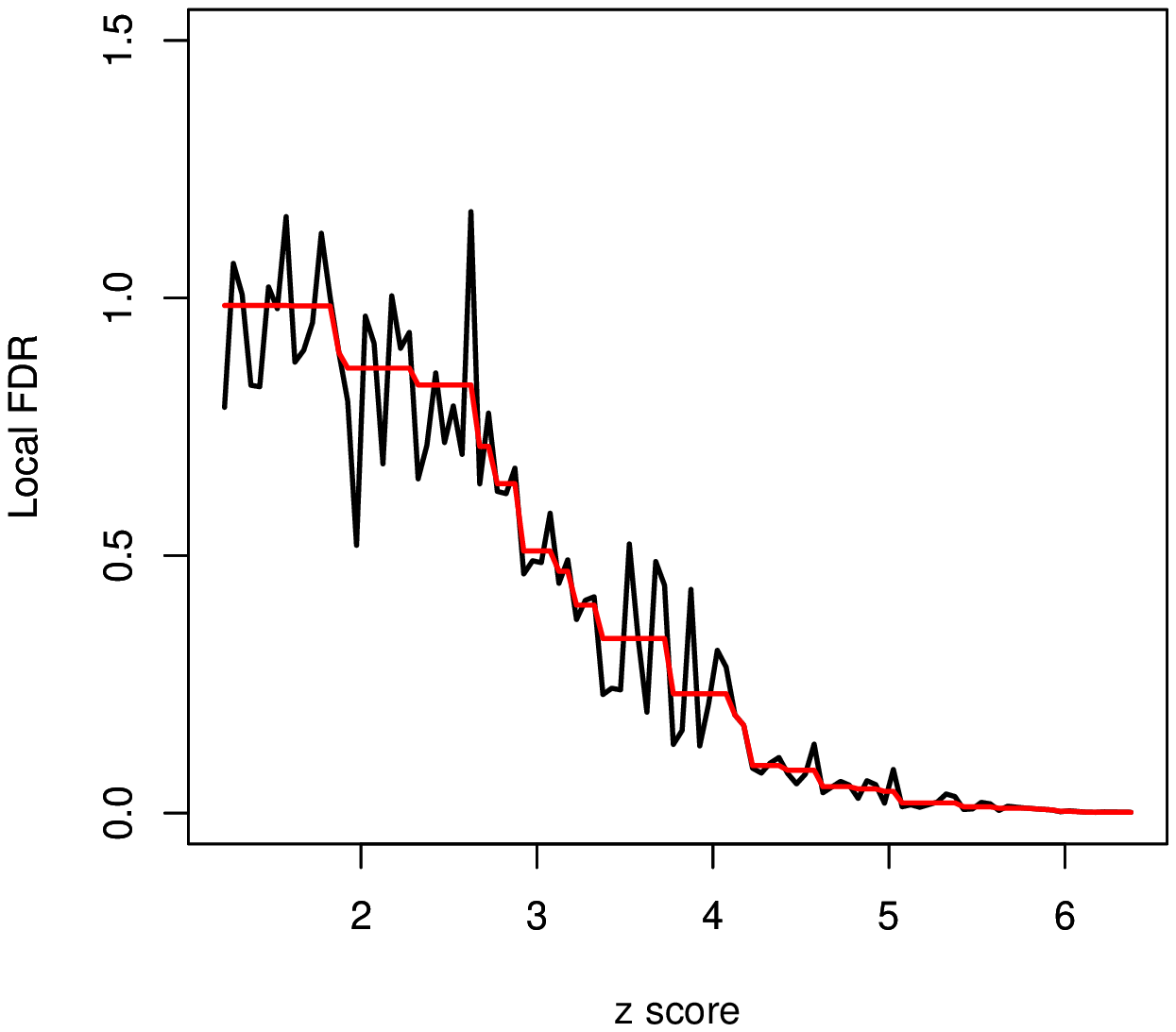,width=0.7\textwidth,height=0.25\textheight}}
  \centerline{(d) Estimated local FDR - right}
\end{minipage}

\caption{The transformed $z$-statistics are from $z_i = \Phi(F_{34}(t_i))^{-1}$, where $F_{34}$ is the CDF of $t$-distribution with degree of 
freedom 34 and $\Phi$ is a standard normal CDF. The estimated local
FDR is obtained by Schwartzman's method \citep{Schwartzman:2008}. }
\label{fig:adeno}
\end{center}
\end{figure}